\documentclass{article}
\usepackage{amsmath,amsthm,amssymb,amscd,amstext,amsfonts}
\usepackage{color,verbatim,graphicx,fullpage,url}
\usepackage[dvipsnames]{xcolor}
\usepackage{algorithm}
\usepackage{algorithmic}
\usepackage{nicefrac}
\usepackage{fullpage,tikz,enumitem}
\usepackage{tikz-cd}
\usepackage{bm}
\usepackage{physics}
\usepackage{todonotes}
\usepackage{mathtools}
\usepackage{thm-restate}

\usepackage[breaklinks=true, linktocpage=true,]{hyperref}
\hypersetup{
    colorlinks = true,
    linkcolor={purple},
    urlcolor={purple},
    citecolor={blue!80!black}
}

\usepackage[nameinlink, noabbrev, capitalize]{cleveref}

\usepackage[utf8]{inputenc}
 
\Crefname{enumi}{Property}{Properties}

\theoremstyle{plain}

\newtheorem{thm}{Theorem}[section]
\newtheorem{theorem}[thm]{Theorem}

\crefname{atheorem}{Theorem}{Theorems}
\Crefname{atheorem}{Theorem}{Theorems}

\newtheorem{conjecture}[thm]{Conjecture}

\newtheorem{corollary}[thm]{Corollary}
\newtheorem{lemma}[thm]{Lemma}
\newtheorem{proposition}[thm]{Proposition}
\newtheorem{definition}[thm]{Definition}

\newtheorem{problem}[thm]{Problem}

\newtheorem*{claim*}{Claim}

\theoremstyle{remark}
\newtheorem{remark}[thm]{Remark}

\renewcommand\norm[1]{\left|\!\left|#1\right|\!\right|}

\newcommand\normF[1]{\left|\!\left|#1\right|\!\right|_{\rm F}}

\newcommand{\inp}[2]{{\left\langle #1,#2 \right\rangle}}            

\def\1{\mathbf{1}} 
\def\0{\mathbf{0}}

\DeclareMathOperator{\Ex}{\mathbb{E}}           
\DeclareMathOperator{\conv}{conv}     
\DeclareMathOperator{\vol}{vol}  
\DeclareMathOperator{\Parity}{Parity}     

\DeclareMathOperator{\RR}{\mathbb{R}}

\newcommand{\cX}{\mathcal X}

\DeclareMathOperator{\sign}{sgn}

\DeclareMathOperator{\signrank}{signrk} 
\DeclareMathOperator{\strank}{srk}

\DeclareMathOperator{\Maj}{Maj}
\DeclareMathOperator{\PG}{PG}
\DeclareMathOperator{\rect}{rc}

\usepackage{ifthen}
\newcommand{\agnorm}[2][]{
	\ifthenelse{\equal{#2}{}}{
		\widetilde{\gamma}_2^{#1}
	}{
		\widetilde{\gamma}_2^{#1}(#2)
	}
}

\DeclareFontFamily{U}{mathx}{}
\DeclareFontShape{U}{mathx}{m}{n}{<-> mathx10}{}
\DeclareSymbolFont{mathx}{U}{mathx}{m}{n}
\DeclareMathAccent{\widecheck}{0}{mathx}{"71}

\newcommand{\cC}{\mathcal{C}}
\newcommand{\cQ}{\mathcal{Q}}

\newcommand{\cT}{\mathcal{T}}

\DeclareMathOperator{\VCdim}{\textsc{vc}}

\DeclareMathOperator{\loss}{loss}                                             

 \DeclareMathOperator{\proj}{proj}

\begin{document}

\title{Lower Bounds for Approximate Sign Rank}

\author{
     Riju Bindua \and Hamed Hatami \and Hasti Karimi \and Robert Robere\thanks{McGill University, \texttt{\{riju.bindua, hamed.hatami, hasti.karimi, robert.robere\}@mcgill.ca}. Hamed Hatami and Robert Robere are supported by NSERC.}
}
 
\date{\today}

\maketitle

\begin{abstract}
We prove new upper and lower bounds on \emph{$\epsilon$-approximate sign-rank}, a relaxation of sign-rank introduced by Chornomaz, Moran, and Waknine \cite{10.1145/3717823.3718213}. 
We prove that every $m \times n$ sign matrix with approximate sign-rank~$d$ contains a monochromatic rectangle of size $d^{-O(d)}m \times d^{-O(d^2)}n$, paralleling classical results for exact sign-rank.  
As an application of our rectangle theorem, we establish a lower bound of $\Omega_{\epsilon,\gamma}(\sqrt{d/\log d})$ on the $\epsilon$-approximate sign-rank of $\gamma$-margin $d$-dimensional half-spaces.
Prior to our work, the only general lower bound technique known for approximate sign-rank yielded lower bounds of strength $\epsilon^{-1} - 1$, which are constant for fixed~$\epsilon$.

A key ingredient in the proof of our monochromatic rectangle theorem is a new geometric theorem on hyperplane avoidance. We show that for any set of $n$ points in general position in~$\mathbb{R}^d$, there exist $d$ subsets, each of size $d^{-O(d)}\, n$, such that no hyperplane simultaneously splits all of them. The proof of the geometric theorem combines the Forster--Barthe \cite{MR1964645, MR1650312} isotropic position theorem, with the Bourgain--Tzafriri \cite{MR890420} restricted invertibility principle.

Next, we study the relationship between approximate sign-rank and VC dimension. 
We prove a lower bound on approximate sign-rank in terms of VC dimension, and exhibit concept classes of VC dimension $2$ with large approximate sign-rank.

Finally, we study the approximate sign-rank of the $2^m \times 2^m$ Hadamard matrix $H_m$.
The sign-rank of $H_m$ is known to be $\Omega(\sqrt{2^m})$ by Forster's classic theorem \cite{MR1964645}.
Contrasting this result, we adapt an argument of Alman and Williams \cite{10.1109/FOCS.2015.18, AlmanW17} to show that the approximate sign-rank of $H_m$ is at most $m^{O(\sqrt{m} \log (1/\epsilon))}$, and hence the Hadamard matrix does not witness polynomial-strength lower bounds for approximate sign-rank.
By using our VC dimension bound, we prove that the approximate sign-rank of $H_m$ is at least $\Omega_\epsilon(m)$.

\end{abstract}

\pagebreak 

\section{Introduction}

If $A \in \{\pm 1\}^{m \times n}$ is an $m \times n$ matrix then the \emph{sign-rank} of $A$, denoted $\signrank(A)$, is the minimum rank of a real matrix $B$ with $\sign(B_{i,j}) = A_{i,j}$ for all entries $i,j$.  
Sign-rank was introduced by Paturi and Simon~\cite{paturi1986probabilistic} in the context of communication complexity, and has become a fundamental quantity of study in theoretical computer science, with further connections to circuit complexity, combinatorics, discrete geometry, and Banach space theory.
One particular domain in which sign-rank has had a strong impact is \emph{learning theory}, where it represents the minimum dimension required to represent a concept class by an arrangement of points and hyperplanes. 
Given a concept class $\cC \subseteq \{\pm 1\}^{\cX}$ over a finite domain~$\cX$, we define  the $\signrank(\cC)$ to be the sign rank of the $|\cC| \times |\cX|$ matrix $A$ defined by $A_{c, x} = c(x)$. 
Equivalently, the sign rank of $\cC$ is
the smallest~$d$ for which there exist embeddings $\{u_c \in \RR^d\}_{c \in \cC}$ and $\{v_x \in \RR^d\}_{x \in \cX}$ satisfying
\begin{equation}
\label{eq:dimension_complexity}
c(x) = \sign \inp{u_c}{v_x}
\quad \text{for all } c \in \cC,\; x \in \cX.
\end{equation}

\paragraph{Approximate sign-rank.}
Since learning inherently tolerates error, requiring the exact linear realization of a concept class $\cC$ in \eqref{eq:dimension_complexity} can be unnecessarily restrictive. To address this, Chornomaz, Moran, and Waknine~\cite{10.1145/3717823.3718213} introduced a natural relaxation of sign-rank called \emph{approximate sign-rank}\footnote{To be distinguished from \emph{probabilistic} sign-rank (cf.~\cite{AlmanW17, KamathMS20}), which is a different measure. See \cref{sec:related-work} for details.}, which tolerates a small amount of classification error in the representation.
To state the definition of approximate sign-rank, we first need the notion of a realizable distribution. 
A distribution over $[n] \times \{\pm 1\}$ is \emph{realizable} by the $i$-th row of a sign matrix $A \in \{\pm 1\}^{m \times n}$ if every pair $(j,b)$ in its support satisfies $b = A_{i,j}$. Such a distribution is naturally determined by a pair $(i,\mu)$, where $i \in [m]$ is the row index and $\mu$ is a distribution over $[n]$; the label of column $j$ is then simply $A_{i,j}$. 

\begin{definition}[Approximate sign-rank~\cite{10.1145/3717823.3718213}]
\label{def:approx_sign_rank}
For $\epsilon \ge 0$, the \emph{$\epsilon$-approximate sign-rank} of a sign matrix $A \in \{\pm 1\}^{m \times n}$, denoted $\signrank_\epsilon(A)$, is the smallest $d$ for which there exist vectors $\{v_j \in \RR^d\}_{j \in [n]}$ such that for every $i \in [m]$ and every distribution $\mu$ over $[n]$, there exists $u_i \in \RR^d$ satisfying
\[
\Pr_{j \sim \mu}\bigl[\sign \inp{u_i}{v_j} \neq A_{i,j}\bigr] \;\le\; \epsilon.
\]
\end{definition}
For a binary concept class $\cC$ over a finite domain $\cX$, we write $\signrank_\epsilon(\cC) \coloneqq \signrank_\epsilon(A)$ where $A$ is the sign matrix with $A_{c,x} = c(x)$. To calibrate \Cref{def:approx_sign_rank}, observe that $\signrank_0(A) = \signrank(A)$, while at the other extreme, $\signrank_\epsilon(A) = 1$ for every $\epsilon \ge 1/2$, since for each row $i$, one of the constant hypotheses $\pm 1$ always achieves error  
at most $1/2$.

The approximate sign-rank can be much smaller than the exact sign-rank. Chornomaz, Moran, and Waknine established a quasipolynomial separation:

\begin{theorem}[{\cite[Theorem 4]{10.1145/3717823.3718213}}]\label{thm:cmw-separation}
    For every integer $d \geq 0$, there is a finite concept class $\mathcal{C}_d$ such that $\signrank_{1/3}(\cC_d) \leq d$, while $\signrank(\cC_d) = d^{\Omega(\log d)}$.
\end{theorem}
On the other hand, despite the close relationship between the two notions, known lower bound techniques for exact sign-rank do not appear to extend to the approximate setting. The only available general lower bound on approximate sign-rank is a Helly-type argument given in~\cite{10.1145/3717823.3718213}, which yields
\[
\signrank_\epsilon(\cC) \;\ge\; \frac{1}{\epsilon} - 1
\]
whenever $\signrank(\cC) \ge \tfrac{1}{\epsilon} - 1$. For any fixed $\epsilon$, this bound is merely a constant.

\subsection{Our Results}

Our main theorem is a new method for proving lower bounds on approximate sign-rank.
Namely, we show that any matrix with low approximate sign-rank must contain a large monochromatic rectangle. This extends classical results of Alon, Pach, Pinchasi, Radoi\v{c}i\'c, and Sharir~\cite{ALON2005310} and Fox, Pach, and Suk~\cite{fox2016polynomial} from exact to approximate sign-rank.

Combining this structural theorem with a result of Frankl and R\"odl~\cite{MR871675}, we obtain the first poly-logarithmic lower bound on the approximate sign-rank of large-margin half-spaces. To our knowledge, this is also the first super-constant lower bound on the approximate sign-rank of any concept class with constant VC dimension, and our bound is loose by a factor of at most $\tilde O(\sqrt{\log n})$, where $n$ is the dimension of the matrix.

We prove similar strength lower bounds on the approximate sign-rank of the $2^m \times 2^m$ Hadamard matrix $H_m$, showing that $\signrank_\epsilon(H_m) \geq \Omega_\epsilon(m)$. 
However, while the sign-rank of $H_m$ is known to be $\signrank(H_m) \geq \sqrt{2^m}$, we adapt a technique of Alman and Williams \cite{AlmanW17} to prove a surprising upper bound of $m^{O(\sqrt{m} \log 1/\epsilon)}$.
Thus, the Hadamard matrix cannot exhibit polynomial-strength lower bounds on approximate sign-rank.

Finally, we prove that approximate sign-rank is at least linear in the VC dimension, and exhibit a strong separation in the reverse direction: there exist matrices of VC dimension two whose approximate sign-rank grows polynomially in the matrix size.

At the core of our approach is a new geometric theorem on hyperplane avoidance, which we believe is of independent interest: given any set of points in general position, one can find a small number of large subsets such that every homogeneous hyperplane avoids at least one of them. The classical Yao--Yao partition~\cite{YaoYao} achieves a similar avoidance property but requires exponentially many parts in the dimension; our theorem reduces this to linearly many, at the cost of not producing a full partition.

\subsubsection{Large Monochromatic Rectangles and Large-Margin Halfspaces}

A key structural property of low-sign-rank matrices is that such matrices necessarily contain large \emph{monochromatic rectangles}, meaning submatrices in which either every entry is $+1$ or every entry is $-1$.
This phenomenon was first established in~\cite{ALON2005310}, where it was shown that every $m \times n$ sign matrix of sign-rank $d$ contains a monochromatic rectangle of size
\[
\frac{m}{2^{d+1}} \times \frac{n}{2^{d+1}}.
\]
A different $\Omega(m) \times \Omega(n)$ lower bound was later obtained by Fox, Pach, and Suk~\cite{fox2016polynomial} using Chazelle’s hyperplane cutting lemma~\cite{MR1194032}, and a third proof was subsequently given in~\cite{MR4494342}. 
In this work, we extend this property to matrices with small approximate sign-rank and show that it holds even for partially defined matrices.

For a partial sign matrix $A \in \{\pm 1, \star\}^{m \times n}$, a distribution over $[n] \times \{\pm 1\}$ is realizable by the $i$-th row of $A$ if it is supported on pairs $(j, A_{i,j})$ with $A_{i,j} \neq \star$. As before, such a distribution is determined by the row index $i$ together with a distribution $\mu$ over the non-$\star$ columns of row $i$, and the approximate sign-rank of $A$ is defined as in \Cref{def:approx_sign_rank}.

\begin{theorem}[Large monochromatic rectangles]
\label{thm:main_rect}
Let $A \in \{\pm 1,\star\}^{m \times n}$ be a partial sign matrix with $\signrank_\epsilon(A) \le d$ for some $\epsilon \in (0,1/2)$. Then there exist subsets $R \subseteq [m]$ and $S \subseteq [n]$ with
\[
|R| \;\ge\; d^{-C_\epsilon \cdot d}\, m
\qquad\text{and}\qquad
|S| \;\ge\; d^{-C_\epsilon \cdot d^2}\, n,
\]
where $C_\epsilon \coloneqq O(\log(1/(1-2\epsilon)))$, such that every non-$\star$ entry of $A_{R \times S}$ has the same sign.
\end{theorem}
Using this theorem, we improve the lower bounds on approximate sign-rank for the concept class of \emph{large-margin halfspaces}, a foundational concept class in learning theory closely related to support vector machines. The formal definition for this concept class over the discrete cube is as follows.
Given a margin parameter $\gamma > 0$, define the class $\cQ_\gamma^d$ 
of partial concepts $c_y \colon \{\pm 1\}^d \to \{\pm 1, \star\}$, indexed by $y \in \{\pm 1\}^d$, by
\[
c_y(x) \;\coloneqq\;
\begin{cases}
\sign\inp{x}{y} & \text{if } |\inp{x}{y}| > \gamma d,\\
\star & \text{otherwise}.
\end{cases}
\]
Informally, the realizable distributions for $\cQ_\gamma^d$ are exactly those supported on points lying at distance at least $\gamma$ from a separating homogeneous hyperplane. 

By applying their Helly-type lower bound, Chornomaz, Moran, and Waknine~\cite[Theorem~6]{10.1145/3717823.3718213} showed that for every $\gamma < 1$ and sufficiently large $d$,
\[
\signrank_\epsilon(\cQ^d_\gamma) \;\ge\; \frac{1}{\epsilon} - 1,
\]
which is constant for any fixed $\epsilon$. We improve this to a bound which is polynomial\footnote{We note that the underlying matrix has $n = 2^d$ rows and columns, and so as a function of the parameter $n$ our lower bound is logarithmic.}  in~$d$, by combining \Cref{thm:main_rect} with a result of Frankl and R\"odl~\cite{MR871675} (see \Cref{thm:FR}) that upper bounds the size of monochromatic rectangles in $\cQ^d_\gamma$.

\begin{theorem}[Lower bound for large-margin half-spaces]
\label{thm:LargeMargin}
For every $\epsilon \in [0,1/2)$ and $\gamma \in (0,1)$,
\[
\signrank_\epsilon(\mathcal{Q}^d_\gamma)
= \Omega_{\epsilon,\gamma}\!\left(\sqrt{\frac{d}{\log d}}\right).
\]
\end{theorem}

A few remarks are in order. 
First, the embedding $x \mapsto x/\sqrt{d}$ maps $\{\pm 1\}^d$ into $\mathbb{S}^{d-1}$, so \Cref{thm:LargeMargin} yields corresponding lower bounds for large-margin half-spaces over the unit sphere. Second, Hatami, Hosseini, and Meng~\cite{MR4617400} showed that $\signrank(\cQ^d_\gamma) = \Omega(d)$ for every $\gamma < 1$. This is tight, since the definition of $\cQ^d_\gamma$ directly provides a $d$-dimensional representation. \Cref{thm:LargeMargin} therefore places $\signrank_\epsilon(\cQ^d_\gamma)$ in the range $\tilde\Omega(\sqrt{d})$ to $d$; we leave the problem of closing this gap to future research.

\paragraph{A spectral lower bound.}
We conclude with a corollary to \Cref{thm:main_rect} that provides a lower bound on approximate sign-rank in terms of the operator norm.

\begin{corollary}
\label{cor:spectral_lower}
For every $A \in \{\pm 1\}^{n \times n}$ and $\epsilon \in [0,1/2)$,
\[
\signrank_\epsilon(A) \;\ge\; \Omega_\epsilon\!\left(\sqrt{\frac{\log(n/\!\norm{A})}{\log\log n}}\right).
\]
\end{corollary}

\begin{proof}
We combine \Cref{thm:main_rect} with a spectral upper bound on the size of monochromatic rectangles. If $S \times T$ is a monochromatic rectangle in $A$, then
\[
|S|^2\,|T|^2 \;=\; \Bigl|\sum_{i \in S,\, j \in T} A_{i,j}\Bigr|^2 \;\le\; \bigl(\norm{A}\,\norm{\mathbf{1}_S}\,\norm{\mathbf{1}_T}\bigr)^2 \;=\; \norm{A}^2\,|S|\,|T|,
\]
which gives $|S|\,|T| \le \norm{A}^2$.  

On the other hand, if $\signrank_\epsilon(A) = d$, then \Cref{thm:main_rect} produces a monochromatic rectangle of size at least $d^{-O_\epsilon(d)}\,n \times d^{-O_\epsilon(d^2)}\,n$, so
\[
d^{-O_\epsilon(d^2)}\,n^2 \;\le\; \norm{A}^2.
\]
Taking logarithms, $\log(n/\!\norm{A}) \le O_\epsilon(d^2 \log d)$, and rearranging gives the claimed bound.
\end{proof}

\subsubsection{The Hadamard matrix}
The Hadamard matrix is the classical example of a matrix with large sign-rank. 
Letting $n = 2^m$, the $n \times n$ Hadamard matrix $H_{m}$ has rows and columns indexed by $\{0,1\}^m$, with entries
\[
H_m(x,y) \;\coloneqq\; (-1)^{\sum_{i=1}^m x_i y_i}.
\]
Forster~\cite{MR1964645} proved $\signrank(H_m) \ge \sqrt{n} = \sqrt{2^m}$, the first polynomial lower bound on the sign-rank of an explicit matrix. This makes the Hadamard matrix a natural candidate for an explicit example with large approximate sign-rank.

However, we show that the approximate sign-rank of $H_m$ is at most sub-polynomial in $n = 2^m$, ruling it out as a source of polynomial lower bounds.

\begin{theorem}
\label{thm:hadamard}
Let $H_m$ be the $n \times n$ Hadamard matrix with $n = 2^m$. For every $\epsilon \in (0,1/2)$,
\[
\Omega_\epsilon\!\left(m\right) \;\le\;\signrank_\epsilon(H_m) \;\le\; m^{O(\sqrt{m}\log(1/\epsilon))}.
\]
\end{theorem}

On the other hand, as shown in \Cref{rem:random_matrices}, a counting argument implies that a random $n \times n$ sign matrix has approximate sign-rank $\Omega_\epsilon(n/\log^2 n)$ with high probability. Thus, matrices with polynomially large approximate sign-rank do exist in principle.

\subsubsection{Approximate Sign-Rank versus VC Dimension}\label{sec:intro_VC}
We next relate approximate sign-rank to the classical notion of \emph{VC dimension} which, by the fundamental theorem of PAC learning, characterizes the sample complexity of PAC-learning~\cite[Theorem 6.7]{shalev2014understanding}. 
Recall that a concept class $\cC \subseteq \{\pm 1\}^{\cX}$ \emph{shatters} a set $S \subseteq \cX$ if
\[
\{ c|_S \mid c \in \cC \} = \{\pm 1\}^S,
\]
where $c|_S$ denotes the restriction of $c$ to $S$. The \emph{VC dimension} of $\cC$ is
\[
\VCdim(\cC) \coloneqq \sup \{ |S| \mid S \subseteq \cX \text{ is shattered by } \cC \}.
\]
For a sign matrix $A$, we define $\VCdim(A)$ as the VC dimension of the associated concept class whose concepts are the rows of $A$ and whose domain elements are the columns.

It is well known~\cite[Theorem 9.2]{shalev2014understanding} that the VC dimension of homogeneous half-spaces in $\mathbb{R}^d$ is $d$  and, as a consequence,
\begin{equation}
\label{eq:signrank_VC}
\VCdim(A) \le \signrank(A).
\end{equation}
For approximate sign-rank, however, the relationship is less direct. In particular, it is not hard to construct (cf.~\Cref{prop:trivial_upper}) sign matrices $A$ for which $\signrank_\epsilon(A) < \VCdim(A)$ for some $\epsilon < 1/2$. This raises the question of whether $\signrank_\epsilon(A)$ can still be lower bounded in terms of $\VCdim(A)$.
Our next result shows that this is indeed the case.
\begin{theorem}
\label{thm:VC_less_asgnrank}
For every $\epsilon \in [0,1/2)$ and every sign matrix $A$,
\[
\signrank_\epsilon(A) \;\ge\; \Omega_\epsilon(\VCdim(A)).
\]
\end{theorem}

While \Cref{thm:VC_less_asgnrank} shows that VC dimension is a lower bound on approximate sign-rank, the two parameters can differ significantly. For exact sign-rank, every sign matrix with VC dimension $1$ has sign-rank at most $3$~\cite[Theorem~2]{DBLP:conf/colt/AlonMY16}, yet Alon, Moran, and Yehudayoff~\cite{DBLP:conf/colt/AlonMY16} exhibited $n \times n$ sign matrices with VC dimension $2$ and sign-rank $\Omega(\sqrt{n}/\log n)$. We prove that a similar separation holds for approximate sign-rank.

\begin{theorem}
\label{thm:alpha_signrank_separation}
For every fixed $\epsilon \in (0,1/2)$, there exist $n \times n$ sign matrices $A$ with
\[
\VCdim(A) \;\le\; 2 \qquad\text{and}\qquad \signrank_\epsilon(A) \;\ge\; \Omega_\epsilon\!\left(\frac{\sqrt{n}}{\log^2 n}\right).
\]
\end{theorem}

\subsection{Technical Overview}

\subsubsection{Monochromatic Rectangles via Hyperplane Avoidance}\label{sec:intro_proof_main}

To prove \Cref{thm:main_rect}, we build on the approach of Alon et al.~\cite{ALON2005310}, who established the analogous result for exact sign-rank. We begin by recalling their argument, and then explain the new ideas needed to handle approximate sign-rank.

We say that a hyperplane \emph{avoids} a set $U \subseteq \RR^d$ if $U$ is 
contained in one of the two open half-spaces bounded by the hyperplane,  or equivalently, if the hyperplane does not intersect the convex hull $\conv(U)$.    

\paragraph{The exact sign-rank~\cite{ALON2005310}.}  The key tool in the argument of Alon et al.~\cite{ALON2005310} is the following classical partition theorem of Yao and Yao.

\begin{theorem}[Yao--Yao~\cite{YaoYao}]
If $\mu$ is a continuous and everywhere-positive density on $\RR^d$, then $\RR^d$ can be partitioned into $2^d$ regions, each of mass $1/2^d$, such that every hyperplane avoids at least one region.
\end{theorem}

Let $A \in \{\pm 1\}^{m \times n}$ with $\signrank(A) = d$, and let $\{u_i\}_{i \in [m]}, \{v_j\}_{j \in [n]} \subseteq \RR^d$ satisfy $A_{i,j} = \sign\inp{u_i}{v_j}$. Applying the Yao--Yao theorem to a smooth approximation of the uniform distribution over $\{v_j\}$ yields a partition $S_1,\ldots,S_{2^d}$ of $\{v_j\}_{j \in [n]}$ into almost-equal-sized parts such that every hyperplane avoids at least one part.

For any row $i \in [m]$, the hyperplane defined by $u_i$ avoids some $S_r$, so all entries $A_{i,j}$ with $j \in S_r$ share the same sign. By the pigeonhole principle over the choice of $r \in [2^d]$ and the sign $b \in \{\pm 1\}$, there exist a set of rows $R \subseteq [m]$ with $|R| \ge m/2^{d+1}$ and a part $S_r$ such that $\sign \inp{u_i}{v_j}=b$ for all $i \in R$ and $j \in S_r$. Therefore, $A_{R \times S_r}$ is a monochromatic submatrix with the desired density.

 \paragraph{The difficulty with approximate sign-rank.}
When $\signrank_\epsilon(A) = d$, the deterministic vectors $u_i$ are replaced by distributions over hyperplanes: By a standard minimax argument (see \Cref{prop:dual_asr}), one can show that there exist vectors $\{v_j\}_{j \in [n]} \subseteq \RR^d$ and, for each row $i$, a distribution $\eta_i$ over $\RR^d$ such that 
\[\Pr_{u \sim \eta_i}[\sign\inp{u}{v_j} \neq A_{i,j}] \le \epsilon \text{ for all } j.\]  
We could apply the Yao--Yao theorem again, and partition $V$ into $S_1, S_2, \ldots, S_{2^d}$.
But now we have a problem: $\eta_i$ is not guaranteed to correctly classify every entry in the $i$th row of $A$, and we can only conclude the weaker property that there is an index $r$ such that the hyperplane distributed according to $\eta_i$ avoids $S_r$ with probability $\geq 1/2^d$.

\paragraph{The density increment argument.}
To get around this issue, we break the argument into two cases. The simplest case is when, for at least half of the rows $i \in [m]$, some part $S_r$ makes the sub-row $A_{\{i\} \times S_r}$ monochromatic. Then the pigeonhole argument from the exact case applies directly, producing a large monochromatic rectangle.

The interesting case is when this fails: for at least half of the rows, every sub-row $A_{\{i\} \times S_r}$ contains both signs. 
Now we use the distributions $\eta_i$ to perform a density increment argument.
The key idea is to introduce a parameter $\beta(A,d)$, defined as the largest $\beta > 0$ such that for every row in the matrix $A$, there is a randomized hyperplane that separates the $+1$ and $-1$ entries with probability at least $\beta$. 
Initially, a union bound shows that if $\signrank_\epsilon(A)=d$, then $\beta(A,d) \ge 1-2\epsilon$.

As we remarked above, with probability $\geq 1/2^d$, the randomized hyperplane distributed by $\eta_i$ avoids some $S_r$, and therefore, does not separate the $+$ and $-$ in $S_r$. 
Hence, by restricting the columns to $S_r$, and conditioning on the event that $\eta_i$ does not avoid $S_r$, we can boost $\beta$ by a factor of $2^d/(2^d - 1)$, at the cost of passing to a submatrix. Since repeating this process $O_\epsilon(2^d)$ times would push $\beta$ above $1$, which is an impossibility,   at some point before that we must arrive at the simple case, producing a large monochromatic rectangle.

\subsubsection{A New Hyperplane Avoidance Theorem}\label{sec:intro_geo}

In the argument sketched above, the Yao--Yao theorem imposes a high price: in each iteration, we restrict to a $ 1/2^d$ fraction of all rows and columns, and then we must repeat this up to $2^d$ times. 
Overall, this yields a \emph{doubly}-exponential loss in density of the final constructed rectangle.
To improve this, we prove the following new hyperplane avoidance theorem, which produces $d$ sets rather than $2^d$ sets.
The proof of \cref{thm:main_rect} then follows the outline from the previous section, with \Cref{thm:main_geo} replacing the Yao--Yao theorem, as now the density-increment process only needs to be repeated $O_\epsilon(d)$ times instead of $O_\epsilon(2^d)$ times.

\begin{theorem}[Hyperplane avoidance]\label{thm:main_geo} 
If $P \subset \mathbb{R}^d$ is a set of points in general position (i.e.~no $d+1$ points lie in a common hyperplane), then there exist subsets
$S_1, \ldots, S_d \subseteq P$ such that
\begin{itemize}
\item $|S_i| \ge d^{-O(d)}|P|$ for all $i=1,\ldots,d$, and
\item every homogeneous hyperplane avoids at least one of      $S_1,\ldots,S_d$.
\end{itemize}
\end{theorem}

\paragraph{Proof overview for \Cref{thm:main_geo}.}
The proof combines two deep ingredients from asymptotic convex geometry: the Forster--Barthe isotropic position theorem~\cite{MR1964645,MR1650312} and the restricted invertibility principle of Bourgain and Tzafriri~\cite{MR890420} (we use the refinement of this principle, introduced by Vershynin~\cite{MR1826503, MR2956233}, but in particular the formulation stated by Marcus, Spielman, and Srivastava \cite{MR4425347}).

By the Forster--Barthe theorem, we may assume without loss of generality that the points of $P$ lie on the unit sphere $\mathbb{S}^{d-1}$ in isotropic position. 
Since we work in the regime where $|P| \gg d$, most points behave as \emph{density points}: a sufficiently small neighbourhood around such a point contains a non-negligible fraction of nearby points of $P$. After discarding the few non-density points, we retain a large subset $P' \subseteq P$ that remains approximately isotropic.

Thanks to the approximate isotropy of the remaining points, we can apply the restricted invertibility theorem to select $d$ representative density points $u_1,\ldots,u_d$ from $P'$ whose position vectors form a $d \times d$ matrix $D$ with large least singular value $\sigma_{\min}(D) = \Omega(d^{-3/2})$. For each $i \in [d]$, we define $S_i$ to be the set of points of $P$ lying within a small spherical cap of radius $\epsilon = d^{-5/2}$ around $u_i$, and the density property guarantees $|S_i| \ge d^{-O(d)}|P|$.

Now, suppose by contradiction that there is a hyperplane defined by a vector $w \in \mathbb{R}^d$ which does not avoid all the sets $S_i$ of points we have constructed.
By the construction of the sets $S_i$, this means that for each row $u_i$, there is a vector $w_i$ near to $u_i$ such that $\langle w, w_i \rangle = 0$ for all $i$.
We can collect these $w_i$ vectors into the rows of another matrix $M$ so that $Mw = 0$. 
However, since $w_i$ and $u_i$ are close for each $i$, it follows that $\norm{D - M}$ must also be small. 
Since $Mw = 0$, the fact that $\norm{D - M}$ is small also implies that $\norm{Dw}$ is small, but this contradicts the fact that $D$ has a large minimum singular value. 
Hence, no such hyperplane $w$ can exist, and the proof is complete.

\subsubsection{Approximate Sign-Rank and VC-Dimension via Counting Arguments}

Finally, we give brief overviews of our proofs of \cref{thm:VC_less_asgnrank} and \cref{thm:alpha_signrank_separation}, both of which use counting arguments. \cref{thm:VC_less_asgnrank} follows from a simple counting argument using the Sauer-Shelah lemma, so we refer to \cref{sec:VC_less_asgnrank} for formal details, and focus on the proof of \cref{thm:alpha_signrank_separation}.

The proof of \cref{thm:alpha_signrank_separation} follows from two components.
The first component comes from prior work of Alon, Moran, and Yehudayoff \cite{DBLP:conf/colt/AlonMY16}, who showed that the number of sign matrices with VC dimension at most $2$ is quite large (at least $2^{\Omega(n^{3/2})}$), and the number of sign matrices with sign rank at most $d$ is relatively small (at most $2^{O(nd \log n)}$). Our separation will follow from counting the number of matrices with approximate sign-rank by reduction to the number of matrices with small sign-rank.

To do this, we derandomize approximate sign-rank by an Adleman-style argument \cite{MR539832}. Namely, we show that if $A$ is a sign matrix with $\signrank_\epsilon(A) \leq d$, then there are $n$ vectors $v_1, \dots, v_n \in \mathbb{R}^d$ and, for each $i \in [m]$, there are $k = O_{\epsilon}(\log n)$ vectors $u_{i, 1}, \dots, u_{i, k} \in \RR^d$ such that $A_{i,j}$ will be the majority vote of $\sign \langle u_{i,\ell}, v_j \rangle$ over each $\ell = 1, \dots, k$.
An easy embedding argument then shows that the number of $n \times n$ matrices with approximate sign rank at most $d$ will be bounded by the number of $nk \times nk$ matrices with sign rank at most $d$, which we can upper bound by the claim of \cite{DBLP:conf/colt/AlonMY16} stated above.

\subsection{Related works and open questions}
\label{sec:related-work}

We now survey the known methods for lower-bounding sign-rank and approximate sign-rank, and highlight what remains open.
We refer to \cite{MR4494342} for a more in-depth survey on lower bound methods for sign-rank.

\paragraph{Lower-bound methods for sign-rank.} 

There are essentially three existing methods for lower-bounding the sign-rank of a matrix $A$: \emph{VC dimension} \cite{paturi1986probabilistic}, \emph{monochromatic rectangle density} \cite{ALON2005310}, and the \emph{operator norm} \cite{MR1964645}.
We have already discussed VC dimension, so let us recall the other two methods.

Given a sign matrix $A \in \{\pm 1\}^{m \times n}$, define its \emph{monochromatic rectangle density} as
\[
\rho(A) \;\coloneqq\; \max_{\substack{R \subseteq [m],\, S \subseteq [n] \\ A_{R \times S} \text{ monochromatic}}} \frac{|R|\,|S|}{mn},
\]
and its \emph{rectangle complexity} as 
\[\rect(A) \coloneqq \max_{R \subseteq [m],\, S \subseteq [n]} \frac{1}{\rho(A_{R \times S})}.\]

For the exact sign-rank, the VC lower bound~\cite{paturi1986probabilistic} and the monochromatic rectangle lower bound of~\cite{ALON2005310} show 
\begin{equation}
\label{eq:signrank_lowers}
\VCdim(A) \;\le\; \signrank(A) \qquad\text{and}\qquad \log\rect(A) \;\le\; \signrank(A).
\end{equation}

Since any $2^d \times d$ sign matrix whose rows include all sign patterns in $\{\pm 1\}^d$ has monochromatic rectangle density $\rho = \max_{k \in [d]} k\, 2^{d-k}/(d\, 2^d) = 1/(2d)$, we always have $2\VCdim(A) \le \rect(A)$, and consequently $\log\rect(A) \ge \log\VCdim(A)$. 
So while the $\log\rect(A)$ bound in~\eqref{eq:signrank_lowers} can be weaker than the VC bound, it is at most exponentially so. 
In the other direction, the VC bound can be far weaker: as shown in~\cite[Theorem~3.2]{MR4494342}, there exist $n \times n$ matrices with $\VCdim = 2$ and $\rect = \Omega(n^2)$, for which the rectangle complexity gives a logarithmic lower bound while the VC dimension gives only a constant. 
The rectangle complexity is thus a qualitatively stronger lower-bound method.  

In \Cref{thm:main_rect,thm:VC_less_asgnrank}, we establish analogues of both bounds in~\eqref{eq:signrank_lowers} for approximate sign-rank. For every $\epsilon \in [0,1/2)$:
\begin{equation}
\label{eq:apx_signrank_lowers}
\Omega_\epsilon\bigl(\VCdim(A)\bigr) \;\le\; \signrank_\epsilon(A) \qquad\text{and}\qquad \Omega_\epsilon\!\left(\sqrt{\frac{\log\rect(A)}{\log\log\rect(A)}}\right) \;\le\; \signrank_\epsilon(A).
\end{equation}

The classic mistake-bound analysis of the Perceptron algorithm~\cite{MR10388,rosenblatt1958perceptron} (see also~\cite[Theorem 9.1]{shalev2014understanding} or \cite[Proposition 17]{alon2022theory}), shows $\VCdim(\cQ_\gamma^d) \le 1/\gamma^2$. Therefore, \Cref{thm:VC_less_asgnrank} can only yield a constant lower bound on $\signrank_\epsilon(\cQ^d_\gamma)$ for fixed $\gamma$, and in particular it cannot recover our super-constant lower bound in \Cref{thm:LargeMargin}.

\paragraph{The logarithmic barrier for approximate sign-rank.} 
Both VC dimension and rectangle complexity lower bounds are inherently limited to $O(\log n)$ for $n \times n$ matrices, since $\VCdim(A) \le \log_2 n$ and $\rect(A) \le n^2$. For the exact sign-rank, this logarithmic barrier stood for nearly two decades until the breakthrough of Forster~\cite{MR1964645}, who used the isotropic position theorem (\Cref{thm:forster}) to prove
\[
\signrank(A) \;\ge\; \frac{n}{\norm{A}},
\]
where $\norm{A}$ denotes the operator norm. For the Hadamard matrix, where Forster's bound gives exact sign-rank $\sqrt{n}$, we show in \Cref{thm:hadamard} that the approximate sign-rank is at most sub-polynomial in $n$. Therefore, unlike in the exact setting, a small operator norm $\norm{A}$ does not by itself imply a polynomial lower bound on approximate sign-rank. Proving a polynomial lower bound for any explicit matrix remains open (such matrices exist by \Cref{rem:random_matrices}).

\begin{problem}
\label{prob:explicit}
Prove a polynomial lower bound in $n$ on $\signrank_\epsilon(A)$ for an explicit family of sign matrices $A \in \{\pm 1\}^{n \times n}$ and a fixed $\epsilon \in (0,1/2)$.  
\end{problem}

\paragraph{Approximate sign-rank vs.~sign-rank.} Another central question is how much smaller the approximate sign-rank can be compared to the exact sign-rank.
We have already mentioned the result of Chornomaz, Moran, and Waknine showing a quasi-polynomial separation between exact and approximate sign-rank (cf.~\cref{thm:cmw-separation}) \cite{10.1145/3717823.3718213}.
The construction in \Cref{thm:cmw-separation} is simple: the domain consists of a set of points in $\mathbb{R}^d$, and each concept is a majority of the signs of three homogeneous half-spaces. The upper bound on $\frac{1}{3}$-approximate sign-rank follows directly from this representation, while the lower bound on exact sign-rank relies on a result showing that the class of intersections of two half-spaces in $\mathbb{R}^d$ (i.e., concepts of the form $x \mapsto \sign(\inp{u_1}{x})\wedge \sign(\inp{u_2}{x})$) have sign-rank at least $d^{\Omega(\log d)}$~\cite[Corollary 1.2]{MR4313288}.

We conjecture that a much stronger separation holds.

\begin{conjecture}
\label{conj:gap}
There exists a fixed constant $d \in \mathbb{N}$ such that there are sign matrices with
\[
\signrank_{1/3}(A)=d
\]
and arbitrarily large $\signrank(A)$.
\end{conjecture}

It is plausible that the same construction underlying \Cref{thm:cmw-separation} already witnesses this conjecture. This is closely related to a major open problem: whether concept classes defined by finite point sets in $\mathbb{R}^d$ and concepts given by intersections of two half-spaces have bounded sign-rank. While the sign-rank is known to be bounded for $d=3$~\cite[Proposition 3.12]{MR4494342}, it remains open in higher dimensions.

\paragraph{Approximate sign-rank vs.~probabilistic sign-rank.}  

Finally, we distinguish the \emph{approximate} sign-rank from the \emph{probabilistic} sign-rank, as introduced by Alman and Williams \cite{AlmanW17}, and further studied in a learning context by Kamath, Montasser, and Srebro \cite{KamathMS20}. 
If $\mathcal{M}$ is a distribution over $m \times n$ sign matrices and $A \in \{\pm 1\}^{m \times n}$, then we say $\mathcal{M}$ \emph{$\epsilon$-represents} $A$ if $\Pr_{M \sim \mathcal{M}}[M_{i,j} \neq  A_{i,j}] \leq \epsilon$ for all $i, j$, and the \emph{sign-rank} of $\mathcal{M}$ is defined to be the maximum sign-rank of any matrix in the support of $\mathcal{M}$.
\begin{definition}[cf.~\cite{AlmanW17}]
    The \emph{$\epsilon$-probabilistic sign-rank} of $A$, denoted $\widetilde{\signrank}_\epsilon(A)$, is the minimum $d$ for which there is a distribution of matrices $\mathcal{M}$ $\epsilon$-representing $A$ with sign-rank $d$.
\end{definition}
Looking ahead, it is perhaps easiest to compare this definition with the dual description of approximate sign-rank that we give in \cref{sec:basic_properties}.
In \cref{prop:dual_asr}, we show that if the approximate sign-rank of $A$ is $d$, then there are vectors $v_1, \dots, v_n \in \RR^d$ and distributions $\eta_1, \eta_2, \ldots, \eta_m$ over $\RR^d$ such that for each $i, j$, \[\Pr_{u \sim \eta_i}[\sign \langle u, v_j \rangle \neq A_{i,j}] \leq \epsilon.\]
Hence, in an approximate sign-rank decomposition of $A$, the vectors $v_j$ for each column $j$ are fixed, while the vectors $\eta_i$ for the rows are randomized, and we must correctly return each entry of the matrix with probability $\geq 1-\epsilon$.
In probabilistic sign-rank, we can randomize the matrix $M$ representing $A$ completely, subject to the condition that the sign-rank of $M$ is at most $d$. From this, we can immediately deduce that 
\[\widetilde{\signrank}_\epsilon(A) \leq \signrank_\epsilon(A),\]
and hence the approximate sign-rank is a restriction of the probabilistic sign-rank.

It is natural to ask how these two measures compare and whether or not they can be separated.
Our bound on the approximate sign-rank of the Hadamard matrix (cf.~\cref{thm:hadamard}) uses the same construction --- via low-degree probabilistic polynomials for symmetric functions \cite{10.1109/FOCS.2015.18} --- that Alman and Williams used \cite{AlmanW17} to show the probabilistic sign-rank of the Hadamard matrix is small in the constant-error regime.
We leave determining the precise relationship between these measures as an open problem:

\begin{problem}
    Determine the asymptotic relationship between probabilistic sign-rank and approximate sign-rank for arbitrary sign matrices.
\end{problem}

\paragraph{Organization.}  \Cref{sec:basic_properties} collects basic properties of approximate sign-rank, including a dual formulation via the minimax theorem. \Cref{sec:Proof_Geom} proves the main geometric theorem (\Cref{thm:main_geo}). \Cref{sec:rect_proof} establishes the monochromatic rectangle theorem (\Cref{thm:main_rect}) and the applications to the the large-margin half-spaces (\Cref{thm:LargeMargin}) and the Hadamard matrix (\Cref{thm:hadamard}). \Cref{sec:VC_proof} proves the VC dimension bounds (\Cref{thm:VC_less_asgnrank} and \Cref{thm:alpha_signrank_separation}).

\section{Basic properties of approximate sign-rank}
\label{sec:basic_properties}

In this short section, we collect a few preliminary observations on approximate sign-rank.

\subsection{A dual formulation}

We work throughout this section with partial sign matrices $A \in \{\pm 1,\star\}^{m \times n}$, where $A_{i,j} = \star$ indicates that the entry is undefined. The \emph{sign-rank} of such a matrix is the smallest $d$ for which there exist vectors $\{v_j \in \RR^d\}_{j \in [n]}$ and $\{u_i \in \RR^d\}_{i \in [m]}$ such that
\[
A_{i,j} = \sign\inp{u_i}{v_j}
\quad \text{for all } (i,j) \text{ with } A_{i,j} \neq \star.
\]
The notions of realizability and approximate sign-rank extend to partial matrices as described in earlier: a distribution over $[n] \times \{\pm 1\}$ is realizable by the $i$-th row of $A$ if it is supported on pairs $(j, A_{i,j})$ with $A_{i,j} \neq \star$, and $\signrank_\epsilon(A)$ is defined as in \Cref{def:approx_sign_rank} with distributions restricted accordingly.

The following reformulation of approximate sign-rank, obtained by swapping quantifiers via the minimax theorem, will be used throughout. Instead of requiring that, for every realizable distribution, there exists a good deterministic classifier, it asks that, for every row, there exists a good \emph{randomized} classifier.

\begin{proposition}[Dual formulation of approximate sign-rank]
\label{prop:dual_asr}
Let $A \in \{\pm 1,\star\}^{m \times n}$ be a partial sign matrix and let $\epsilon \ge 0$. Then $\signrank_\epsilon(A)$ is the smallest $d$ for which there exist vectors $\{v_j \in \RR^d\}_{j \in [n]}$ such that for every $i \in [m]$, there exists a finitely supported distribution $\eta_i$ over $\RR^d$ satisfying
\[
\Pr_{u \sim \eta_i}\bigl[\sign \inp{u}{v_j} \neq A_{i,j}\bigr] \;\le\; \epsilon
\quad \text{for all } j \in [n] \text{ with } A_{i,j} \neq \star.
\]
\end{proposition} 

\begin{proof}
The set of distinct labellings of $[n]$ of the form $j \mapsto \sign\inp{u}{v_j}$ as $u$ ranges over $\mathbb{R}^d$
is finite. Hence, there exists a finite set $P \subseteq \mathbb{R}^{d}$ that realizes all such labellings, and it suffices in both formulations to restrict $u$ to $P$.
For a finite set $K$, let $\Delta(K)$ denote the set of probability distributions on $K$. Fix a row index $i \in [m]$, and let
\[
N_i \;\coloneqq\; \{j \in [n] : A_{i,j} \neq \star\}
\]
be the set of columns on which the $i$-th row is defined. Any distribution over $[n] \times \{\pm 1\}$ realizable by the $i$-th row of $A$ is equivalent to a distribution over $N_i$, with labels determined by the $i$-th row. Thus, \Cref{def:approx_sign_rank} requires
\[
\max_{\mu \in \Delta(N_i)} \;\min_{u \in P} \;
\Ex_{j \sim \mu}\bigl[\1\bigl(\sign\inp{u}{v_j} \neq A_{i,j}\bigr)\bigr]
\;\le\; \epsilon,
\]
while the present formulation requires
\[
\min_{\eta \in \Delta(P)} \;\max_{j \in N_i} \;
\Ex_{u \sim \eta}\bigl[\1\bigl(\sign\inp{u}{v_j} \neq A_{i,j}\bigr)\bigr]
\;\le\; \epsilon.
\]
By von Neumann's minimax theorem, these two conditions are equivalent.
\end{proof}

\subsection{Approximate sign-rank can be very small}

The following proposition shows that approximate sign-rank collapses to a constant as soon as the error parameter approaches $1/2$.

\begin{proposition}
\label{prop:trivial_upper}
For every sign matrix $A \in \{\pm 1\}^{m \times n}$,\;
$\signrank_{1/2 - 1/(2n)}(A) \le 2$.
\end{proposition}

\begin{proof}
Choose distinct points $v_1,\ldots,v_n \in \mathbb{S}^1$ with no two antipodal. Fix a row index $i \in [m]$ and a distribution $\mu$ on $[n]$. Let $j$ be an atom of maximum mass under $\mu$; since $\mu$ is supported on at most $n$ atoms, $\mu(j) \ge 1/n$.

Let $w \in \mathbb{S}^1$ be orthogonal to $v_j$. Since no two of the $v_{j'}$ are antipodal, $\inp{w}{v_{j'}} \neq 0$ for every $j' \neq j$. For sufficiently small $\delta > 0$, the vectors
\[
u \;\coloneqq\; w + \delta A_{ij}\, v_j
\qquad\text{and}\qquad
u' \;\coloneqq\; -w + \delta A_{ij}\, v_j
\]
both satisfy $\sign\inp{u}{v_j} = \sign\inp{u'}{v_j} = A_{ij}$, while $\sign\inp{u}{v_{j'}} = -\sign\inp{u'}{v_{j'}}$ for every $j' \neq j$. In particular, both $u$ and $u'$ correctly classify $(j, A_{ij})$, and for every other atom exactly one of the two classifies it correctly. Therefore
\[
\loss_{i,\mu}(u) + \loss_{i,\mu}(u') \;=\; 1 - \mu(j) \;\le\; 1 - \tfrac{1}{n},
\]
and the better of the two achieves population loss at most $\tfrac{1}{2} - \tfrac{1}{2n}$.
\end{proof}

The dimensional collapse in \Cref{prop:trivial_upper} requires discontinuous dependence of the classifier $u$ on the distribution $\mu$. The following proposition shows that enforcing continuity could force the dimension back up to $n$.

\begin{proposition}[Borsuk--Ulam obstruction for continuous classifiers]
\label{prop:borsuk_ulam}
Let $\cC = \{\pm 1\}^n$ and $\epsilon < 1/2$. If $\signrank_\epsilon(\cC) \le d$ and the map selecting the classifier can be taken to be continuous, then $d \ge n$.
\end{proposition}

\begin{proof}
Suppose $\signrank_\epsilon(\cC) \le d$, and let $\{v_j \in \RR^d\}_{j \in [n]}$ be the corresponding embedding. Write $\Delta_\cC$ for the set of all distributions on $[n] \times \{\pm 1\}$ realizable by $\cC$. By definition, for every $\mu \in \Delta_\cC$ there exists $\psi(\mu) \in \RR^d$ such that
\[
\loss_\mu\!\bigl(j \mapsto \sign \inp{\psi(\mu)}{v_j}\bigr) \;\le\; \epsilon.
\]

Suppose for contradiction that $\psi\colon\Delta_\cC \to \RR^d$ is continuous and $d < n$. Identify each $f \in \mathbb{S}^{n-1}$ with the distribution $\mu_f \in \Delta_\cC$ that assigns mass $|f_i|/\norm{f}_1$ to the labelled example $(i,\,\sign(f_i))$ for each $i \in [n]$. This map is a homeomorphism between $\mathbb{S}^{n-1}$ and $\Delta_\cC$, so $f \mapsto \psi(\mu_f)$ is a continuous map from $\mathbb{S}^{n-1}$ to $\RR^d$. Since $d < n$, the Borsuk--Ulam theorem~(see~\cite{matousek2003borsuk}) produces $f \in \mathbb{S}^{n-1}$ with $\psi(\mu_f) = \psi(\mu_{-f})$. Writing $u$ for this common value, the hypothesis $h\colon j \mapsto \sign\inp{u}{v_j}$ satisfies $\loss_{\mu_f}(h) \le \epsilon$ and $\loss_{\mu_{-f}}(h) \le \epsilon$. But $\mu_f$ and $\mu_{-f}$ assign the same masses to the same domain points with every label negated, so $\loss_{\mu_f}(h) + \loss_{\mu_{-f}}(h) = 1$, contradicting $\epsilon < 1/2$.
\end{proof}

\section{Proof of the hyperplane avoidance theorem}
\label{sec:Proof_Geom}

This section is devoted to the proof of \Cref{thm:main_geo}. \Cref{subsec:prelim_geom} collects the geometric and analytic tools we will need: the Forster--Barthe isotropic position theorem, the restricted invertibility principle, and a basic estimate for the number of density points. The proof itself appears in \Cref{subsec:proof_geom}.

\subsection{Preliminaries}
\label{subsec:prelim_geom}

\paragraph{Notation.} We denote the singular values of a matrix $B \in \RR^{m\times d}$ by $\sigma_1(B) \ge \sigma_2(B) \ge \cdots \ge \sigma_{\min(m,d)}(B) \ge 0$. The \emph{Frobenius norm} and \emph{operator norm} of $B$ are
\[
\normF{B} \coloneqq \sqrt{\sum_{i=1}^{m}\sum_{j=1}^{d} B_{ij}^2}
= \sqrt{\sum_{i} \sigma_i(B)^2},
\qquad
\norm{B} \coloneqq \max_{x \ne 0} \frac{\norm{Bx}}{\norm{x}} = \sigma_1(B),
\]
respectively, where $\norm{\cdot}$ on vectors denotes the Euclidean norm. We also write
\[
\sigma_{\min}(B) \;\coloneqq\; \min_{x \ne 0} \frac{\norm{Bx}}{\norm{x}}
\]
for the least singular value of $B$. For $S \subseteq [m]$, we write $B_S$ for the submatrix of $B$ consisting of the rows indexed by $S$.

\paragraph{Isotropic position.}

\begin{definition}[Isotropic position]
Vectors $v_1,\ldots,v_n \in \RR^d$ are in \emph{isotropic position} if
\[
\frac{1}{n}\sum_{i=1}^n \inp{u}{v_i}^2 \;=\; \frac{1}{d}
\qquad\text{for every } u \in \mathbb{S}^{d-1}.
\]
\end{definition}

The following result, implicit in the work of Barthe~\cite{MR1650312} and established independently by Forster~\cite{MR1964645}, shows that any spanning set of vectors can be placed in isotropic position by an invertible linear transformation followed by normalization. Forster used this fact to obtain the first linear lower bound on the sign-rank of an explicit matrix.

\begin{theorem}[\cite{MR1650312,MR1964645}]
\label{thm:forster}
Given $v_1,\ldots,v_n \in \RR^d$ in general position, there exists an invertible $T\colon\RR^d\to\RR^d$ such that the normalized vectors $\{Tv_i/\|Tv_i\|\}_{i=1}^n$ lie in isotropic position on $\mathbb{S}^{d-1}$.
\end{theorem}

We record the following elementary consequence of isotropy, which asserts that there are several points that are far away from any given homogeneous hyperplane.  

\begin{lemma}
\label{lemma:large_inner_product}
Let $v_1,\ldots,v_n \in \mathbb{S}^{d-1}$ be in isotropic position. Then for every $w \in \mathbb{S}^{d-1}$,
\[
\left|\left\{i \in [n] \mid |\inp{v_i}{w}| \ge \tfrac{1}{\sqrt{2d}}\right\}\right|
> \frac{n}{2d}.
\]
\end{lemma}

\begin{proof}
Let $T = \{i \mid |\inp{v_i}{w}| \ge 1/\sqrt{2d}\}$. Isotropy gives $\sum_{i=1}^n \inp{v_i}{w}^2 = n/d$, while splitting the sum according to $T$ yields
\[
\frac{n}{d} = \sum_{i=1}^n \inp{v_i}{w}^2 \;<\; |T| + \frac{n-|T|}{2d} \leq |T| + \frac{n}{2d}.
\]
Rearranging yields $|T| > n/(2d)$.
\end{proof}

\paragraph{Restricted invertibility.}

The restricted invertibility principle of Bourgain and Tzafriri~\cite{MR890420} is a quantitative refinement of the assertion that the rank of an $m \times d$ matrix equals the maximum number of linearly independent rows it contains. Define the \emph{stable rank} of $B \in \RR^{m\times d}$ as 
\[
\strank(B) \;\coloneqq\; \frac{\normF{B}^2}{\norm{B}^2},
\]
which satisfies $\strank(B) \le \rank(B)$. The principle asserts that when $\strank(B)$ is large, there exists a large subset of rows $S$ for which $\sigma_{\min}(B_S)$ is bounded away from zero.

The original formulation in~\cite{MR890420} produced $|S|=\Omega(m)$; sharp quantitative forms valid for every $|S| < \strank(B)$ were subsequently established by Vershynin~\cite{MR1826503}, Spielman and Srivastava~\cite{MR2956233}, and Marcus, Spielman, and Srivastava~\cite{MR4425347}. We will use the following form of Vershynin's bound, as stated in~\cite{MR4425347}.

\begin{theorem}[{\cite[Theorem 1.1]{MR4425347}}]
\label{thm:MSS}
Let $B \in \RR^{m\times d}$ and let $k \le \strank(B)$ be a positive integer. Then there exists $S \subseteq [m]$ of size $k$ with
\[
\sigma_{\min}(B_S)^2 \;\ge\; \left(1 - \sqrt{\tfrac{k}{\strank(B)}}\right)^{2} \frac{\normF{B}^2}{m}.
\]
\end{theorem}

\begin{remark}
\label{rem:MSS}
When the rows of $B$ are unit vectors in isotropic position, $B^{\top}B = (m/d)I_d$, so $\strank(B) = \rank(B) = d$ and all $d$ singular values of $B$ equal $\sqrt{m/d}$. For $k<d$, \Cref{thm:MSS} produces $k$ rows whose submatrix has $\sigma_{\min}^2 \ge (1-\sqrt{k/d})^2$. The bound is vacuous at $k=d$, and to recover a non-trivial lower bound on $\sigma_{\min}$ in full dimension, we will apply \Cref{thm:MSS} with $k=d-1$ and append one further row via the following linear-algebraic lemma.
\end{remark}

\begin{lemma}
\label{lemma:min_singular_value}
Let $\delta, s \in (0,1]$, let $B \in \RR^{(d-1)\times d}$ satisfy $\sigma_{\min}(B) \ge s$, and let $V_B$ denote its row space. Given a unit vector $u \in \RR^d$ with $\|\proj_{V_B^\perp} u\| \ge \delta$, let $A \in \RR^{d\times d}$ be the matrix obtained by appending $u$ as an additional row to $B$. Then
\[
\sigma_{\min}(A) \;\ge\; \frac{\delta s}{4}.
\]
\end{lemma}
 \begin{proof}
Since $s>0$, $B$ has full row rank, so $\dim V_B = d-1$ and $V_B^\perp$ is one-dimensional. Let $x \in \RR^d$ be a unit vector, and write
\[
x = x_1 + x_2, \qquad u = u_1 + u_2,
\]
with $x_1, u_1 \in V_B$ and $x_2, u_2 \in V_B^\perp$. Since $V_B^\perp$ is one-dimensional and $\norm{u_2}\ge\delta$, we have 
\[|\inp{u_2}{x_2}| = \norm{u_2}\norm{x_2} \ge \delta\norm{x_2},\] 
and hence
\[
|\inp{u}{x}| \ge \max\bigl(\delta\norm{x_2} - \norm{x_1},\, 0\bigr).
\]
Since $Bx = Bx_1$,
\begin{equation}
\label{eq:lower_Ax}
\norm{Ax}^2 = \norm{Bx_1}^2 + \inp{u}{x}^2 \ge s^2\norm{x_1}^2 + \max\bigl(\delta\norm{x_2} - \norm{x_1},\,0\bigr)^2.
\end{equation}
We distinguish two cases. If $\norm{x_1} \ge (\delta/2)\norm{x_2}$, then $1 = \norm{x_1}^2 + \norm{x_2}^2 \le (1+4/\delta^2)\norm{x_1}^2$, so $\norm{x_1} \ge \delta/\sqrt{4+\delta^2} \ge \delta/4$, and \eqref{eq:lower_Ax} gives 
\[\norm{Ax} \ge s\norm{x_1} \ge \frac{s\delta}{4}.\] 

Otherwise $\norm{x_1} < (\delta/2)\norm{x_2}$, so $\norm{x_2}\ge 1/2$, and \eqref{eq:lower_Ax} gives
\[
\norm{Ax} \ge \delta\norm{x_2} - \norm{x_1} \ge \tfrac{\delta}{2}\norm{x_2} \ge \tfrac{\delta}{4} \ge \tfrac{s\delta}{4}.
\]
Minimizing over unit $x$ yields $\sigma_{\min}(A) \ge s\delta/4$.
\end{proof}
 
\paragraph{Density points.} Given a finite set $P \subseteq \mathbb{S}^{d-1}$, we say that a point $v \in P$ is an \emph{$(\epsilon,k)$-density point} of $P$ if at least $k$ points of $P$ lie within Euclidean distance $\epsilon$ of $v$.

\begin{lemma}
\label{lemma:density_point}
For any finite $P \subseteq \mathbb{S}^{d-1}$, any $\epsilon \in (0,1]$, and any positive integer $k$, fewer than $k\,(3/\epsilon)^d$ points of $P$ fail to be $(\epsilon,k)$-density points.
\end{lemma}
\begin{proof}
Let $T \subseteq P$ be the set of non-$(\epsilon,k)$-density points, and let $M \subseteq T$ be a maximal $\epsilon$-separated subset of $T$. Since the open balls of radius $\epsilon/2$ centered at points of $M$ are pairwise disjoint and contained in the ball of radius $1+\epsilon/2$, we have
\[
|M| \cdot \vol(B_{\epsilon/2}^d) \;\le\; \vol(B_{1+\epsilon/2}^d),
\]
which gives $|M| \le (3/\epsilon)^d$. By maximality, every point of $T$ lies within Euclidean distance $\epsilon$ of some point of $M$. Since each point of $M$ is a non-$(\epsilon,k)$-density point, fewer than $k$ points of $P$ lie within distance $\epsilon$ of it, so
\[
|T| \;<\; k\,|M| \;\le\; k\,(3/\epsilon)^d.\qedhere
\]
\end{proof}

\subsection{Proof of \texorpdfstring{\Cref{thm:main_geo}}{Hyperplane Avoidance Theorem}}
\label{subsec:proof_geom}
Set $\epsilon \coloneqq d^{-5/2}$. By \Cref{thm:forster}, we may apply an invertible linear transformation and normalize to assume without loss of generality that the points of $P = \{v_1,\ldots,v_n\}$ lie in isotropic position on $\mathbb{S}^{d-1}$. Since the transformation is invertible and the normalization is positive, neither the general position nor the avoiding behaviour of homogeneous hyperplanes is affected.

Let $B \in \RR^{n \times d}$ be the matrix with rows $v_1,\ldots,v_n$. Isotropy gives $B^{\top}B = \tfrac{n}{d} I_d$.

Set $k \coloneqq \tfrac{n}{8d^2} (3/\epsilon)^{-d}=d^{-O(d)}n$. By \Cref{lemma:density_point}, the set $T$ of indices of non-$(\epsilon,k)$-density points in $P$ satisfies 
\[ |T| < \frac{n}{8d^2}.\]

\medskip\noindent\textbf{Stable rank after discarding non-density points.}
Let $\bar{T} \coloneqq [n] \setminus T$ and write $B_{\bar{T}}$ for the submatrix of $B$ consisting of the rows indexed by $\bar{T}$. Define $E \coloneqq \sum_{i \in T} v_i v_i^{\top}$, so that
\[
B_{\bar{T}}^{\top} B_{\bar{T}} \;=\; \tfrac{n}{d}\,I_d - E.
\]
Since $E \succeq 0$ and each $v_i$ is a unit vector,
\[
\norm{E} \;\le\; \Tr(E) \;=\; |T| \;<\; \frac{n}{8d^2}.
\]
Setting $x \coloneqq d\norm{E}/n < 1/(8d)$ and using the inequality $(1-x)/(1+x) \ge 1-2x$, we obtain
\[
\strank(B_{\bar{T}})
\;=\; \frac{n - \Tr(E)}{(n/d) + \norm{E}}
\;\ge\; \frac{n - d\norm{E}}{(n/d) + \norm{E}}
\;=\; d \cdot \frac{1-x}{1+x}
\;\ge\; d(1-2x)
\;>\; d - \tfrac{1}{4}.
\]

\medskip\noindent\textbf{Applying restricted invertibility.} Since $d-1 < \strank(B_{\bar{T}})$, \Cref{thm:MSS} produces a subset $S \subseteq \bar{T}$ of size $d-1$. Let $A' \coloneqq B_S$. Since every row of $B_{\bar{T}}$ is a unit vector, $\normF{B_{\bar{T}}}^2 = |\bar{T}|$, and so
\[
\sigma_{\min}(A')^2
\;\ge\; \left(1 - \sqrt{\frac{d-1}{d - 1/4}}\right)^{2}
\;=\; \left(1 - \sqrt{1 - \frac{3/4}{d-1/4}}\right)^{2}.
\]
Applying $1 - \sqrt{1-t} \ge t/2$ for $t \in [0,1]$ gives
\[
s \;\coloneqq\; \sigma_{\min}(A') \;\ge\; \frac{3}{8(d-1/4)} \;=\; \Omega(1/d).
\]

\medskip\noindent\textbf{Appending a $d$-th density point.} Let $V'$ denote the row space of $A'$ and let $w$ be a unit vector perpendicular to $V'$. By \Cref{lemma:large_inner_product}, at least $n/(2d)$ vectors in $P$ satisfy $|\inp{v_i}{w}| \ge 1/\sqrt{2d}$. Since $|T| < n/(8d^2) \le n/(2d)$, at least one such vector lies in $\bar{T}$; call it $u_d$. In particular, $u_d$ is an $(\epsilon,k)$-density point and
\[
\norm{\proj_{V'^{\perp}}(u_d)} \;\ge\; |\inp{u_d}{w}| \;\ge\; \frac{1}{\sqrt{2d}} \;\eqqcolon\; \delta.
\]
Let $D \in \RR^{d \times d}$ be the matrix obtained by appending $u_d$ as the $d$-th row of $A'$. By \Cref{lemma:min_singular_value},
\[
\sigma_{\min}(D) \;\ge\; \frac{\delta s}{4} \;=\; \Omega(d^{-3/2}).
\]

\paragraph{Final construction.} Denote the rows of $D$ by $u_1,\ldots,u_d$. For each $i \in [d]$, define
\[
S_i \;\coloneqq\; \{v \in P \mid \norm{v-u_i} \le \epsilon\}.
\]
Since every $u_i$ is an $(\epsilon,k)$-density point, $|S_i| \ge k = n \cdot d^{-O(d)}$.

Suppose, for contradiction, that some $w \in \mathbb{S}^{d-1}$ defines a homogeneous hyperplane that does not avoid any $S_i$. Then for each $i \in [d]$, the hyperplane meets $\conv(S_i)$, and consequently, there exists $w_i \in B_\epsilon(u_i)$ with $\inp{w}{w_i}=0$. 

Let $M \in \RR^{d\times d}$ be the matrix with rows $w_1,\ldots,w_d$. Since $Mw = 0$ and $w \ne 0$, the matrix $M$ is singular. On the other hand, since $\norm{w_i-u_i} \le \epsilon$, we have 
\[
\norm{D - M}^2 \;\le\; \normF{D-M}^2 \;=\; \sum_{i=1}^{d} \norm{u_i - w_i}^2 \;\le\; d\epsilon^2 \;=\; d^{-4}.
\]
For any  $w \in \mathbb{S}^{d-1}$, we have 
\[\norm{Mw} \ge \norm{D w + (M-D)w} \ge \norm{D w} -  \norm{(M-D)w} \ge \sigma_{\min}(D)-\norm{M-D}  \ge \Omega(d^{-3/2}) - d^{-2} > 0,  \]
for $d$ sufficiently large, contradicting the singularity of $M$.

\section{Proof of the monochromatic rectangle theorem}
\label{sec:rect_proof}
 
This section is devoted to the proof of \Cref{thm:main_rect}.
After proving \cref{thm:main_rect}, we use it first to prove \cref{thm:LargeMargin} by exploiting a result of Frankl and R\"{o}dl \cite{MR871675}, and then to prove our bounds on the approximate sign-rank of the Hadamard matrix in \cref{thm:hadamard}.
We work throughout this section with partial sign matrices $A \in \{\pm 1,\star\}^{m \times n}$, where $A_{i,j} = \star$ indicates that the entry is undefined. The notions of realizability and approximate sign-rank extend to partial matrices as described in \Cref{sec:basic_properties}.

\paragraph{The separation parameter.} The following parameter is central to the proof. It quantifies, for a given matrix and dimension, how effectively a randomized hyperplane can simultaneously separate every pair of oppositely labelled entries in each row.

\begin{definition}
\label{def:beta}
Let $A \in \{\pm 1,\star\}^{m \times n}$ be a partial sign matrix and let $d \in \mathbb{N}$. Define $\beta(A,d)$ to be the supremum of $\beta \ge 0$ for which there exist vectors $\{v_j\}_{j \in [n]} \subseteq \mathbb{R}^{d}$ and finitely supported distributions $\eta_1,\ldots,\eta_m$ on $\mathbb{R}^{d}$ such that for every $i \in [m]$ and every $j_1, j_2 \in [n]$ with $A_{i,j_1} = -1$ and $A_{i,j_2} = +1$,
\[
\Pr_{w \sim \eta_i}\bigl[\sign\inp{w}{v_{j_1}} = -1 \;\text{and}\; \sign\inp{w}{v_{j_2}} = +1\bigr] \;\ge\; \beta.
\]
\end{definition}

Note that $\beta(A,d)$ is a property of the matrix $A$ and the dimension $d$ alone: the definition optimizes over all embeddings and all distributions. The approximate sign-rank assumption provides a non-trivial initial lower bound.

\begin{lemma}
\label{lem:beta_initial}
If $\signrank_\epsilon(A) \le d$, then $\beta(A,d) \ge 1 - 2\epsilon$.
\end{lemma}

\begin{proof}
By the dual formulation (\Cref{prop:dual_asr}), there exist vectors $\{v_j\}_{j \in [n]} \subseteq \mathbb{R}^{d}$ and, for each $i \in [m]$, a finitely supported distribution $\eta_i$ over $\mathbb{R}^{d}$ such that
\[
\Pr_{w \sim \eta_i}\bigl[\sign\inp{w}{v_j} \neq A_{i,j}\bigr] \;\le\; \epsilon
\quad \text{for all } j \in [n] \text{ with } A_{i,j} \neq \star.
\]
Fix $i \in [m]$ and $j_1, j_2$ with $A_{i,j_1} = -1$ and $A_{i,j_2} = +1$. By a union bound,
\[
\Pr_{w \sim \eta_i}\bigl[\sign\inp{w}{v_{j_1}} \neq -1 \;\text{or}\; \sign\inp{w}{v_{j_2}} \neq +1\bigr]
\;\le\; 2\epsilon,
\]
so the complementary event has probability at least $1 - 2\epsilon$.
\end{proof}

\paragraph{The key inductive step.} The next lemma contains the central argument of the proof: given any matrix $A$ with $\beta(A,d) \ge \beta$, we either find a large monochromatic rectangle in $A$, or pass to a large submatrix where $\beta$ has strictly increased.  

\begin{lemma}[Monochromatic rectangle or $\beta$-boost]
\label{lem:monochromatic_or_boost}
Let $A \in \{\pm 1,\star\}^{m \times n}$ with $\beta(A,d) \ge \beta$ for some $\beta > 0$. Then at least one of the following holds:
\begin{enumerate}
\item[\textup{(a)}] There exist $R \subseteq [m]$ and $S \subseteq [n]$ with $|R| \ge m/(4d)$ and $|S| \ge d^{-O(d)}\, n$ such that every non-$\star$ entry of $A_{R \times S}$ has the same sign.
\item[\textup{(b)}]  There exist $R \subseteq [m]$ and $S \subseteq [n]$ with $|R| \ge m/(2d)$ and $|S| \ge d^{-O(d)}\, n$ such that
\[
\beta(A_{R \times S},\, d) \;\ge\; \frac{d}{d-1}\,\beta.
\]
\end{enumerate}
\end{lemma}

\begin{proof}
Let $\{v_j\}_{j \in [n]} \subseteq \mathbb{S}^{d-1}$ and $\{\eta_i\}_{i \in [m]}$ be an embedding and distributions witnessing $\beta(A,d) \ge \beta$. Apply \Cref{thm:main_geo} to the point set $\{v_j\}_{j \in [n]}$, obtaining index sets $S_1,\ldots,S_d \subseteq [n]$ such that
\begin{enumerate}
\item $|S_r| \ge d^{-O(d)}\, n$ for every $r \in [d]$, and
\item every homogeneous hyperplane avoids at least one of $\{v_j\}_{j \in S_1},\ldots,\{v_j\}_{j \in S_d}$.
\end{enumerate}

Call a row $i \in [m]$ \emph{easy} if all non-$\star$ entries of $A_{\{i\} \times S_r}$ share the same sign for some $r \in [d]$, and \emph{hard} otherwise.

\paragraph{Case 1: at least $m/2$ rows are easy.}
By pigeonhole over $r \in [d]$ and $b \in \{\pm 1\}$, at least $m/(4d)$ easy rows share the same part $S_r$ and the same sign $b$. These rows and $S_r$ form a monochromatic rectangle, giving~(a).

\paragraph{Case 2: at least $m/2$ rows are hard.}
Every hard row $i$ has the property that $A_{\{i\} \times S_r}$ contains both $+1$ and $-1$ entries for every $r \in [d]$. By the avoidance property, every $w \in \mathbb{S}^{d-1}$ avoids at least one of $\{v_j\}_{j \in S_1},\ldots,\{v_j\}_{j \in S_d}$. Averaging over $w \sim \eta_i$,
\[
\sum_{r=1}^{d} \Pr_{w \sim \eta_i}[w \text{ avoids } \{v_j\}_{j \in S_r}] \;\ge\; 1,
\]
so for each hard row $i$ there exists $r_i \in [d]$ with $\Pr_{w \sim \eta_i}[w \text{ avoids } \{v_j\}_{j \in S_{r_i}}] \ge 1/d$. By pigeonhole over $r \in [d]$, at least $m/(2d)$ hard rows share the same index $r$; call this set $R$.

It remains to show $\beta(A_{R \times S_r}, d) \ge \frac{d}{d-1}\,\beta$. Fix $i \in R$ and $j_1, j_2 \in S_r$ with $A_{i,j_1} = -1$ and $A_{i,j_2} = +1$. Since $v_{j_1}$ and $v_{j_2}$ both belong to $\{v_j\}_{j \in S_r}$, any $w$ achieving $\sign\inp{w}{v_{j_1}} = -1$ and $\sign\inp{w}{v_{j_2}} = +1$ necessarily fails to avoid $\{v_j\}_{j \in S_r}$. Letting 
\[E_r \coloneqq \{w \mid w \text{ does not avoid } \{v_j\}_{j \in S_r}\},\] 
our choice of $r$ gives $\Pr_{w \sim \eta_i}[E_r] \le (d-1)/d$. Conditioning on $E_r$,
\[
\Pr_{w \sim \eta_i}\bigl[\sign\inp{w}{v_{j_1}} = -1 \;\text{and}\; \sign\inp{w}{v_{j_2}} = +1 \;\big|\; E_r\bigr]
\;\ge\; \frac{\beta}{(d-1)/d}
\;=\; \frac{d}{d-1}\,\beta.
\]
Since this holds for every $i \in R$ and every oppositely labelled pair in $S_r$, the restricted embedding $\{v_j\}_{j \in S_r}$ together with the conditional distributions $\{\eta_i(\cdot \mid E_r)\}_{i \in R}$ witness $\beta(A_{R \times S_r}, d) \ge \frac{d}{d-1}\,\beta$, giving~(b).
\end{proof}

\paragraph{Iterating the inductive step.} We iterate \Cref{lem:monochromatic_or_boost}. Set $A^{(0)} \coloneqq A$ and $\beta^{(0)} \coloneqq 1 - 2\epsilon$, which is a valid initial bound by \Cref{lem:beta_initial}. At step $t$, we apply \Cref{lem:monochromatic_or_boost} to $A^{(t)}$ with $\beta = \beta^{(t)}$. If a monochromatic rectangle is found, we terminate. Otherwise, we pass to a submatrix $A^{(t+1)}$ with
\[
\beta^{(t+1)} \;\ge\; \frac{d}{d-1}\,\beta^{(t)},
\]
while the number of rows decreases by a factor of at most $4d$ and the number of columns decreases by a factor of at most $d^{O(d)}$.

Since $\beta \le 1$, boosting $\beta$ can occur at most $T$ times, where $T=O\!\left(d\log\frac{1}{1-2\epsilon}\right)$ is the smallest integer with $(d/(d-1))^T(1-2\epsilon) > 1$. 

Therefore, a monochromatic rectangle must be found at some step $t < T$. Tracking the losses across all iterations, the final rectangle $A_{R \times S}$ satisfies
\[
|R| \;\ge\; (4d)^{-T}\, m \;=\; d^{-O(d\log(1/(1-2\epsilon)))}\, m
\]
and
\[
|S| \;\ge\; d^{-O(d) \cdot T}\, n \;=\; d^{-O(d^2\log(1/(1-2\epsilon)))}\, n,
\]
as claimed. 

\paragraph{Approximate sign-rank of large-margin halfspaces.}
We can combine \cref{thm:main_rect} with the following known result, stating that sufficiently large submatrices of $\mathcal{Q}^n_\gamma$ are non-monochromatic, to prove \cref{thm:LargeMargin}:
\begin{theorem}[Corollary of Theorem 1.5 in~\cite{MR871675}]
\label{thm:FR}
For every $\gamma \in (0,1)$, there exists $\delta=\delta(\gamma)>0$ such that the following holds. For every $S,T \subseteq \{\pm 1\}^d$ with
\[
|S|\,|T| > (4-\delta)^d,
\]
there exist $(x,y),(x',y') \in S \times T$ such that
\[
\langle x,y \rangle > \gamma d
\qquad\text{and}\qquad
\langle x',y' \rangle < -\gamma d.
\]
\end{theorem}

\begin{proof}[Proof of \cref{thm:LargeMargin}]
   Let $k = \signrank_\epsilon(\cQ^d_\gamma)$.
   Applying \cref{thm:main_rect} we obtain that $\cQ^d_\gamma$ has a monochromatic submatrix indexed by $R \times S$ with \[|R||S| \geq k^{-C_\epsilon k} \cdot k^{-C_\epsilon k^2} 4^k = k^{-C_\epsilon(k + k^2)} \cdot 4^d.\]
   On the other hand, by the theorem above, we must have $|R||S| < (4-\delta)^d$. Combining the two inequalities, letting $\alpha = 4/(4-\delta)$, and rearranging yields \[\alpha^d \leq k^{C_\epsilon(k+k^2)} \leq k^{C_\epsilon(2k^2)}.\]
   Taking logs yields $k^2 \log k = \Omega_{\epsilon, \gamma}(d)$, and hence $k = \Omega_{\epsilon, \gamma}(\sqrt{d/\log d})$ by a routine calculation.
\end{proof}

\paragraph{Approximate sign-rank of the Hadamard matrix.}
We now prove \Cref{thm:hadamard}. The lower bound follows from our spectral corollary to our monochromatic rectangle theorem. The upper bound adapts the argument of Alman and Williams~\cite{AlmanW17} for the probabilistic sign-rank of $H_m$. 
The key observation to extend the argument of \cite{AlmanW17} to approximate sign-rank is that in the resulting rank factorization, the column vectors $v_y$ can be chosen independently of the distribution over hyperplanes.

\begin{proof}[Proof of \Cref{thm:hadamard}]
The lower bound is immediate from \Cref{thm:VC_less_asgnrank} and the fact that $\VCdim(H_m)=m$: the columns indexed by the standard basis vectors $e_1,\ldots,e_m \in \{0,1\}^m$ are shattered, since $H_m(x,e_i) = (-1)^{x_i}$, which ranges over all sign patterns as $x$ varies over $\{0,1\}^m$.

For the upper bound, recall that the Hadamard matrix $H_m$ is defined by $H_m(x,y) = (-1)^{\sum_i x_i y_i}$ for $x,y \in \{0,1\}^m$, so that 
\[H_m(x,y) = \Parity(x_1 y_1, \ldots, x_m y_m),\] where $\Parity(z) \coloneqq (-1)^{|z|}$ is the sign-valued parity function.

Alman and Williams~\cite{10.1109/FOCS.2015.18} showed that for every symmetric function $f\colon \{0,1\}^m \to \mathbb{Z}$, there exists a distribution $\mu$ over multilinear polynomials of degree $d \coloneqq O(\sqrt{m}\log(1/\epsilon))$ with integer coefficients such that
\[
\Pr_{p \sim \mu}[f(z) \neq p(z)] \;\le\; \epsilon
\qquad \text{for all } z \in \{0,1\}^m.
\]
Applying this to $\Parity$, we obtain a distribution $\mu$ over degree-$d$ polynomials with
\[
\Pr_{p \sim \mu}[H_m(x,y) \neq p(x_1 y_1, \ldots, x_m y_m)] \;\le\; \epsilon
\qquad \text{for all } x,y \in \{0,1\}^m.
\]

We now express each polynomial evaluation as an inner product. Let $\cT \coloneqq \{S \subseteq [m] : |S| \le d\}$. Writing $p(z) = \sum_{S \in \cT} a_S \prod_{i \in S} z_i$, we have
\[
p(x_1 y_1, \ldots, x_m y_m)
\;=\; \sum_{S \in \cT} a_S \prod_{i \in S} x_i \prod_{i \in S} y_i
\;=\; \inp{u_{x,p}}{v_y},
\]
where $v_y \coloneqq \bigl[\prod_{i \in S} y_i\bigr]_{S \in \cT} \in \RR^{|\cT|}$ and $u_{x,p} \coloneqq \bigl[a_S \prod_{i \in S} x_i\bigr]_{S \in \cT} \in \RR^{|\cT|}$. Crucially, $v_y$ does not depend on the polynomial $p$.

For each $x \in \{0,1\}^m$, let $\eta_x$ be the distribution over $\RR^{|\cT|}$ induced by sampling $p \sim \mu$ and forming $u_{x,p}$. Then for every $x,y \in \{0,1\}^m$,
\[
\Pr_{u \sim \eta_x}\bigl[\sign\inp{u}{v_y} \neq H_m(x,y)\bigr] \;\le\; \epsilon.
\]
By the dual formulation (\Cref{prop:dual_asr}),
\[
\signrank_\epsilon(H_m) \;\le\; |\cT| \;=\; \binom{m}{\le d} \;\le\; m^{O(\sqrt{m}\log(1/\epsilon))}.\qedhere
\]
\end{proof}

\section{VC dimension bounds for approximate sign-rank}
 \label{sec:VC_proof}
 In this section, we present the proofs of \Cref{thm:VC_less_asgnrank} and \Cref{thm:alpha_signrank_separation}.

\subsection{Proof of \texorpdfstring{\Cref{thm:VC_less_asgnrank}}{VC-Dimension Bound}} 
\label{sec:VC_less_asgnrank}
 \begin{proof}[Proof of \Cref{thm:VC_less_asgnrank}] 
 Let $\VCdim(A) = n$. Then $A$ contains a $2^n \times n$ submatrix $B$ whose rows are all $2^n$ distinct sign vectors in $\{\pm 1\}^n$. Since $\signrank_\epsilon(B) \le \signrank_\epsilon(A)$, it suffices to show
\[
d \;\coloneqq\; \signrank_\epsilon(B) \;=\; \Omega_\epsilon(n).
\]
If $d \ge n$, there is nothing to prove, so assume $d < n$.

Since $\signrank_\epsilon(B) = d$, there exist vectors $\{v_j\}_{j \in [n]} \subseteq \RR^d$ such that for every $i \in [2^n]$, there exists $u_i \in \RR^d$ satisfying
\begin{equation}
\label{eq:approx_sign_rank}
\bigl|\{j \in [n] : \sign\inp{u_i}{v_j} \neq B_{i,j}\}\bigr| \;\le\; \epsilon n.
\end{equation}
(This is the definition of approximate sign-rank applied to the uniform distribution over $[n]$.)

We use \eqref{eq:approx_sign_rank} to upper-bound the number of distinct rows of $B$. Since the VC dimension of half-spaces in $\RR^d$ is $d$, the Sauer--Shelah lemma gives at most $\binom{n}{\le d}$ distinct functions of the form $j \mapsto \sign\inp{u}{v_j}$ as $u$ ranges over $\RR^d$. For each such function, \eqref{eq:approx_sign_rank} allows at most $\binom{n}{\le \epsilon n}$ choices for the error locations. The number of distinct rows of $B$ is therefore at most $\binom{n}{\le d}\binom{n}{\le \epsilon n}$. Since $B$ has $2^n$ rows,
\[
\binom{n}{\le d}\binom{n}{\le \epsilon n} \;\ge\; 2^n.
\]
Using the standard bound $\binom{n}{\le k} \le 2^{n\,h(k/n)}$ for $k \le n/2$ (see, e.g.,~\cite[Lemma~16.19]{Flum2006ParameterizedCT}), where $h$ denotes the binary entropy function, this gives
\[
h(d/n) + h(\epsilon) \;\ge\; 1.
\]
Since $\epsilon < 1/2$, we have $h(\epsilon) < 1$, and therefore $h(d/n) \ge 1 - h(\epsilon) = \Omega_\epsilon(1)$, which implies $d = \Omega_\epsilon(n)$.
\end{proof}

\subsection{Proof of \texorpdfstring{\Cref{thm:alpha_signrank_separation}}{Approximate Sign-Rank Separation}}

The following lemma is due to Alon, Moran, and Yehudayoff~\cite[Lemma~22]{DBLP:conf/colt/AlonMY16}.

\begin{lemma} 
\label{lem:signrank_counting}
The number of $n \times n$ sign matrices with sign-rank at most $d$ is at most $2^{O(n d\log(n))}$.
\end{lemma}

 The following is implicit in~\cite{DBLP:conf/colt/AlonMY16}, but we include the proof for completeness. 

\begin{lemma} 
\label{lem:VC2_count}
The number of $n \times n$ sign matrices with $\VCdim \le 2$ is at least $2^{\Omega(n^{3/2})}$.
\end{lemma}

\begin{proof}
Let $n = m^2 + m + 1$, and let $\PG(2,m)$ be the projective plane of order $m$, which contains $n$ points and $n$ lines with each line incident to exactly $m+1$ points. Define the $n \times n$ sign matrix $A$ by setting $A_{l,p} = +1$ if point $p$ lies on line $l$, and $A_{l,p} = -1$ otherwise. The matrix $A$ contains exactly $n(m+1) = \Theta(n^{3/2})$ entries equal to $+1$.

Since any two distinct lines of $\PG(2,m)$ meet in exactly one point, $A$ contains no $2 \times 2$ all-$1$'s submatrix. This property is preserved if we change any subset of the $+1$ entries to $-1$, so we obtain $2^{\Theta(n^{3/2})}$ distinct sign matrices, none containing a $2 \times 2$ all-$1$'s submatrix. Since any matrix with $\VCdim \ge 3$ must contain such a submatrix, each of these matrices has VC dimension at most $2$.
\end{proof}

\paragraph{Derandomization of approximate signrank.}  The following lemma is  analogous to Adleman's theorem~\cite{MR539832} in complexity theory. 

\begin{lemma}
\label{lem:Adleman}
Let $A \in \{\pm 1\}^{m \times n}$ be a sign matrix and let $\epsilon \in (0,1/2)$. If $\signrank_\epsilon(A) \le d$, then there exist vectors $\{v_j \in \RR^d\}_{j \in [n]}$ and an odd integer $k = O_\epsilon(\log n)$ such that for every $i \in [m]$, there exist $u_{i,1},\ldots, u_{i,k} \in \RR^d$ with
\[
A_{i,j} \;=\; \Maj\bigl(\sign\inp{u_{i,1}}{v_j},\;\ldots,\;\sign\inp{u_{i,k}}{v_j}\bigr)
\quad \text{for all } j \in [n],
\]
where $\Maj\colon \{\pm 1\}^k \to \{\pm 1\}$ denotes the majority function.
\end{lemma}

\begin{proof}
By the dual formulation (\Cref{prop:dual_asr}), there exist vectors $\{v_j \in \RR^d\}_{j \in [n]}$ such that for every $i \in [m]$, there is a finitely supported distribution $\eta_i$ over $\RR^d$ with
\[
\Pr_{u \sim \eta_i}\bigl[\sign\inp{u}{v_j} \neq A_{i,j}\bigr] \;\le\; \epsilon
\quad \text{for all } j \in [n].
\]
Fix $i \in [m]$ and draw $u_{i,1},\ldots,u_{i,k}$ independently from $\eta_i$. For a fixed $j \in [n]$, let 
\[X \; \coloneqq\; \sum_{l=1}^{k} \1[\sign\inp{u_{i,l}}{v_j} \neq A_{i,j}]\]
count the number of misclassifications. Since $\Ex[X] \le k\epsilon$, the majority vote errs on $A_{i,j}$ only if $X \ge k/2$. By Hoeffding's inequality,
\[
\Pr\bigl[X \ge k/2\bigr] \;\le\; e^{-2k(1/2-\epsilon)^2}.
\]
A union bound over $j \in [n]$ gives failure probability at most $n\, e^{-2k(1/2-\epsilon)^2}$, which is strictly less than $1$ for $k = O_\epsilon(\log n)$. Hence, there exists a deterministic choice of $u_{i,1},\ldots,u_{i,k}$ for which the majority vote correctly computes $A_{i,j}$ for all $j$.
\end{proof}

We are now ready to prove \Cref{thm:alpha_signrank_separation}

\begin{proof}[Proof of \Cref{thm:alpha_signrank_separation}]
We show by a counting argument that most $n \times n$ sign matrices with $\VCdim \le 2$ have large approximate sign-rank.

Suppose $A \in \{\pm 1\}^{n \times n}$ satisfies $\signrank_\epsilon(A) \le d$. By \Cref{lem:Adleman} with $k = O_\epsilon(\log n)$, there exist vectors $\{v_j\}_{j \in [n]} \subseteq \RR^d$ and, for each $i \in [n]$, vectors $u_{i,1},\ldots,u_{i,k} \in \RR^d$ such that each entry of $A$ is recovered as a majority vote. Define the $nk \times n$ sign matrix $B$ by $B_{(i-1)k+l,\,j} \coloneqq \sign\inp{u_{i,l}}{v_j}$. By construction, $\signrank(B) \le d$, and $A$ is uniquely determined by $B$ via majority votes over consecutive blocks of $k$ rows.

The number of such matrices $B$ is at most the number of $nk \times nk$ sign matrices with sign-rank at most $d$, which by \Cref{lem:signrank_counting} is at most $2^{O(nkd\log(nk))}$. Since $B$ determines $A$, the number of $n \times n$ matrices with $\signrank_\epsilon \le d$ is also at most $2^{O(nkd\log(nk))}$.

By \Cref{lem:VC2_count}, there are at least $2^{\Omega(n^{3/2})}$ sign matrices with $\VCdim \le 2$. For some such matrix to have $\signrank_\epsilon > d$, it suffices that
\begin{equation}
\label{eq:counting_approx_srnk}
2^{O(nkd\log(nk))} \;<\; 2^{\Omega(n^{3/2})},
\end{equation}
which gives $d = \Omega\bigl(\sqrt{n}/(k\log n)\bigr) = \Omega_\epsilon\bigl(\sqrt{n}/\log^2 n\bigr)$.
\end{proof}

The same counting argument yields a lower bound for random matrices.

\begin{remark}[Approximate sign-rank of random matrices]
\label{rem:random_matrices}
Since the total number of $n \times n$ sign matrices is $2^{n^2}$, the bound \eqref{eq:counting_approx_srnk} implies that for a uniformly random $A \in \{\pm 1\}^{n \times n}$,
\[
\Pr\bigl[\signrank_\epsilon(A) \le d\bigr] \;\le\; \frac{2^{O(nkd\log(nk))}}{2^{n^2}} \;=\; o(1)
\]
provided $d \le c_\epsilon\, n/\log^2 n$ for a sufficiently small constant $c_\epsilon > 0$. In other words, the approximate sign-rank of a random $n \times n$ sign matrix is $\Omega_\epsilon(n/\log^2 n)$ with high probability.
\end{remark}

\bibliographystyle{amsalpha}  
\bibliography{refs} 
\end{document}